\documentclass[sigconf,edbt]{acmart-edbt2021}

\def\BibTeX{{\rm B\kern-.05em{\sc i\kern-.025em b}\kern-.08em
    T\kern-.1667em\lower.7ex\hbox{E}\kern-.125emX}}

\usepackage{booktabs} 

\usepackage{amsmath}
\usepackage{setspace}
\usepackage{placeins}
\usepackage{afterpage}
\usepackage[utf8]{inputenc}
\usepackage{graphicx}
\usepackage{amsthm}
\usepackage{color}
\usepackage{caption}
\usepackage{tabu}
\usepackage{array}
\usepackage{booktabs}
\usepackage{threeparttable}
\usepackage{relsize}
\usepackage{physics}

\usepackage{algorithm}
\usepackage[noend]{algpseudocode}

\newlength{\maxwidth}
\newcommand{\algalign}[2]
{\makebox[\maxwidth][r]{$#1{}$}${}#2$}

\usepackage{float}
\usepackage{stfloats}
\usepackage{lipsum}
\usepackage[english]{babel}

\usepackage{comment}

\usepackage[T1]{fontenc}
\usepackage{mwe}    
\usepackage{subfig}
\let\oldnl\nl
\newcommand{\nonl}{\renewcommand{\nl}{\let\nl\oldnl}}
\newtheorem{thm}{Theorem}

\theoremstyle{definition}

\def\Equal{\texttt{=}}

\setcopyright{rightsretained}

\acmDOI{}

\acmISBN{978-3-89318-084-4}

\acmConference[EDBT 2021]{24th International Conference on Extending Database Technology (EDBT)}{March 23-26, 2021}{Nicosia, Cyprus} 
\acmYear{2021}

\begin{document}
\title{An Efficient and Secure Location-based Alert Protocol using Searchable Encryption and Huffman Codes}

\author{Sina Shaham}
\affiliation{%
  \institution{University of Southern California}
  \state{Los Angeles} 
  \country{USA}
}
\email{sshaham@usc.edu}

\author{Gabriel Ghinita}
\affiliation{%
  \institution{University of Massachusetts}
  \state{Boston} 
  \country{USA}
}
\email{gghinita@cs.umb.edu}

\author{Cyrus Shahabi}
\affiliation{%
  \institution{University of Southern California}
  \state{Los Angeles} 
  \country{USA}
}
\email{shahabi@usc.edu}

\renewcommand{\shortauthors}{}

\begin{abstract}
Location data are widely used in mobile apps, ranging from location-based recommendations, to social media and navigation. A specific type of interaction is that of {\em location-based alerts}, where mobile users subscribe to a service provider (SP) in order to be notified when a certain event occurs nearby. 
Consider, for instance, the ongoing COVID-19 pandemic, where contact tracing has been singled out as an effective means to control the virus spread. Users wish to be notified if they came in proximity to an infected individual. However, serious privacy concerns arise if the users share their location history with the SP in plaintext. 

To address privacy, recent work proposed several protocols that can securely implement location-based alerts. The users upload their encrypted locations to the SP, and the evaluation of location predicates is done directly on ciphertexts. When a certain individual is reported as infected, all matching ciphertexts are found (e.g., according to a predicate such as ``10 feet proximity to any of the locations visited by the infected patient in the last week''), and the corresponding users notified. However, there are significant performance issues associated with existing protocols. The underlying searchable encryption primitives required to perform the matching on ciphertexts are expensive, and without a proper encoding of locations and search predicates, the performance can degrade a lot. In this paper, we propose a novel method for variable-length location encoding based on Huffman codes. By controlling the length required to represent encrypted locations and the corresponding matching predicates, we are able to significantly speed up performance. We provide a theoretical analysis of the gain achieved by using Huffman codes, and we show through extensive experiments that the improvement compared with fixed-length encoding methods is substantial.
\end{abstract}

%
%



\maketitle

\section{Introduction}
Location-based alerts are an emerging area of mobile apps that are very relevant to domains such as public safety, healthcare and transportation. For instance, users may want to subscribe to services that notify them whether an imminent danger exists in their close proximity (e.g., an active shooter situation). Or, in the recent context of COVID-19, mobile users wish to be notified if they came in close proximity to an individual who was diagnosed with the disease. While the advantages of location-based alerts are undeniable, they also introduce serious privacy concerns: in order to benefit from such services, users periodically upload their locations to a service provider (SP). The SP monitors large number of individuals, and evaluates spatial predicates to determine which individuals should be alerted. Disclosing individual locations can leak sensitive personal details to the SP, which may in turn share the data with third parties. And even in cases where the SP is fully trusted, it can be subject to cyber-attacks, or subpoenas by governments, resulting in the users' moving history being exposed.

Movement data can disclose sensitive details about an individual's health status, political orientation, alternative lifestyles, etc. Therefore, it is crucial to support location-based alerts while at the same time protecting user privacy. This problem has been recently studied in literature, and formulated in the context of {\em secure alert zones}~\cite{ghinita2014efficient,nguyen2019privacy,shaham2020enhancing}, where users report their {\em encrypted} locations to the SP, and the SP evaluates alert predicates on encrypted data. These approaches require a special kind of encryption that allows predicate evaluation on ciphertexts, namely {\em searchable encryption (SE)} \cite{song2000,boneh2007conjunctive,HXT18}. However, the SE primitives are not specifically designed for geospatial queries, but rather for arbitrary keyword or wildcard queries. As a result, a data mapping step must transform spatial queries to the primitive operations supported on ciphertexts. Due to this translation, the performance overhead can be significant. 

Some solutions use {\em Symmetric Searchable Encryption (SSE)} \cite{song2000,curtmola2011searchable,HXT18}, where a trusted entity knows the secret key of the transformation, and collects the location of all users before encrypting them and sending the ciphertext to the service provider. While the performance of SSE can be quite good, the system model that requires mobile users to share their cleartext locations with a trusted service is not adequate from a privacy perspective, since it still incurs a significant amount of disclosure.

Prior work in secure alert zones~\cite{ghinita2014efficient,nguyen2019privacy,shaham2020enhancing} uses {\em Hidden Vector Encryption (HVE)}~\cite{boneh2007conjunctive}, which is an {\em asymmetric} type of encryption that allows direct evaluation of predicates on top of ciphertexts. Each user encrypts her own location using the {\em public} key of the transformation, and no trusted component that accesses locations in clear is required. However, the performance overhead of HVE in the spatial domain remains high. 

In existing HVE work for geospatial data~\cite{ghinita2014efficient},~\cite{nguyen2019privacy}, the data domain is partitioned into a hierarchical data structure, and each node in this structure is assigned a binary string identifier. The binary representation of each node plays an important part in query encoding, and it influences the amount of computation that needs to be executed when evaluating predicates on ciphertexts. In~\cite{ghinita2014efficient}, the earliest solution for secure alert zones, the impact of the specific encoding is not evaluated in-depth. In~\cite{shaham2020enhancing}, the geospatial data domain is embedded to a high-dimensional hypercube, and then graph embedding~\cite{chandrasekharam1994genetic} is applied to reduce the computation overhead in the predicate evaluation step.

However, all previous solutions use fixed-length encoding of locations and alert zones, meaning that the same number of bits is used to represent each location. In cases where the distribution of alert zones and/or locations is not uniform, using fixed-length encoding can introduce unnecessary overhead. Motivated by this fact, we propose techniques to reduce the computational overhead of HVE by using variable-length encoding.  Specifically, we use Huffman compression codes to represent both user locations and alert predicates. Areas of the domain that are more popular, or more likely to result in a secure alert being triggered, are encoded with fewer bits than less popular areas. This allows us to perform spatial query execution on ciphertexts in a less computationally-intensive manner.

Our specific contributions are:

\begin{itemize}
 
\item We consider for the first time the use of variable-length encoding, specifically Huffman compression codes, for the problem of secure alert zones on encrypted location data;

\item We devise specialized domain encoding techniques for both user locations and alert zones that take into account location popularity;

\item We provide algorithms to evaluate the secure alert zone enclosure predicates directly on ciphertexts when both user locations and alert zones are represented using variable-length encoding;

\item We perform an extensive experimental evaluation which shows that the proposed approach reduces considerably the performance overhead of secure alert zones compared to fixed-length encoding approaches.
    
\end{itemize}

The rest of the paper is organized as follows: Section 2 introduces necessary background and the system model. 
Section 3 provides the details of the proposed variable-length encoding techniques for user locations and alert zones. Section 4 generalizes our solution to non-binary identifiers. Section 5 analyzes the overhead of variable-length encoding on ciphertext size. 
Section 6 provides a security discussion, followed by evaluation of the proposed approach on both real and synthetic datasets in Section 7. We survey related work in Section 8 and conclude in Section 9.

\begin{figure*}[t]
\centering
\includegraphics[scale=.52]{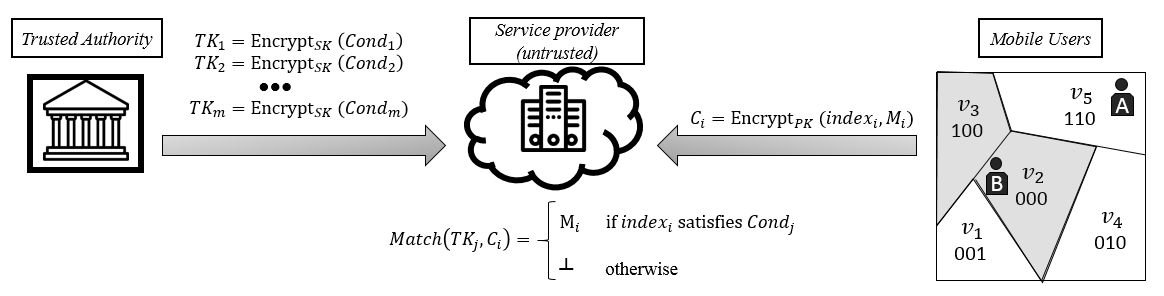}
\hspace{1em}
\centering
\vspace{-18pt}
\caption{Location-based alert system.}
\label{Fig: System model}
\vspace{-15pt}
\end{figure*}

\begin{figure}[t]
	\subfloat[Match]{%
	\includegraphics[scale=.35]{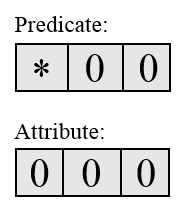}
	}
	\qquad 
	\qquad
	\subfloat[Nonmatch]{%
	\includegraphics[scale=.35]{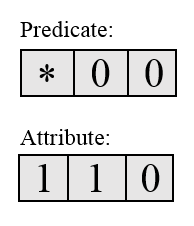}
	}\vspace{-10pt}
	\caption{HVE evaluation}
	\vspace{-15pt}
	\label{Fig: Matching Process}
\end{figure}

\section{Background}\label{Background} 

Consider a map divided into a set of $n$ non-overlapping partitions 
\begin{equation}
    \mathcal{V}=\{ v_1,\,,v_2,..., v_{n} \}.
\end{equation}
Each partition $v_i$ represents a spatial area on the map referred to as {\em cell}. Cells are identified by a unique binary code called {\em index}, and can have arbitrary shapes and sizes (although equal-size square cells are most likely in practice). We refer to the partitioning as a {\em grid}. The assignment of indexes to cells is referred to as {\em grid encoding}. All indexes must have the same length for security purposes (to prevent an adversary from distinguishing cells based on length). Fig.~\ref{Fig: System model} shows a sample grid with five cells, each associated with an index of length three. 

When an event of interest occurs, an {\em alert zone} is created, which spans a number of grid cells. We refer to such cells interchangeably as {\em alert cells} or {\em alerted cells}. In Fig.~\ref{Fig: System model}, cells $v_3$ and $v_2$ associated with the indexes 100 and 000 (shown highlighted) are alert cells. We denote the likelihood of cell $v_i$ being alerted by $p(v_i)$, or alternatively $p_i$. Our goal is to exploit alert cell likelihoods in order to choose an encoding that reduces the computational complexity of HVE. 

\subsection{Hidden Vector Encryption}
\label{sec:app}

{\em Hidden Vector Encryption (HVE)} \cite{boneh2007conjunctive} is a searchable encryption system that supports predicates in the form of conjunctive equality, range and subset queries. Search on ciphertexts can be performed with respect to a number of {\em index attributes}. HVE represents an attribute as a bit vector (each element has value $0$ or $1$), and the search predicate as a {\em pattern} vector where each element can be $0$, $1$ or '*' that signifies a wildcard (or ``don't care'') value. Let $l$ denote the HVE {\em width}, which is the bit length of the attribute, and consequently that of the search predicate. A predicate evaluates to $True$ for a ciphertext $C$ if the attribute vector $I$ used to encrypt $C$ has the same values as the pattern vector of the predicate in all positions that are not '*' in the latter. Fig.~\ref{Fig: Matching Process} illustrates the two cases of {\em Match} and {\em Non-Match} for HVE.

HVE is built on top of a symmetrical bilinear map of composite order \cite{boneh2007conjunctive}, which is a function $e : \mathbb{G} \times \mathbb{G} \rightarrow \mathbb{G}_T$ such that $\forall a,b \in G$ and $ \forall u,v \in \mathbb{Z}$ it holds that $e(a^u,b^v)=e(a,b)^{uv}$. $\mathbb{G}$ and $\mathbb{G}_T$ are cyclic multiplicative groups of composite order $N=P\cdot Q$ where $P$ and $Q$ are large primes of equal bit length. We denote by $\mathbb{G}_p$, $\mathbb{G}_q$ the subgroups of $\mathbb{G}$ of orders $P$ and $Q$, respectively. Let $l$ denote the HVE {\em width}, which is the bit length of the attribute, and consequently that of the search predicate. HVE consists of the following phases:

{\bf Setup.} The public/private ($PK$/$SK$) key pair has the form:
$$SK = ( g_q \in \mathbb{G}_q,\quad a \in \mathbb{Z}_p,\quad \forall i \in [1..l]: u_i,h_i, w_i, g, v \in \mathbb{G}_p )$$
To generate $PK$, we first choose at random elements \(R_{u,i}, R_{h,i}\), \(R_{w,i} 
\in \mathbb{G}_q, \forall i \in [1..l]\)  and \(R_v \in \mathbb{G}_q\). Next, $PK$ is determined as:
$$PK = (g_q,\quad V=vR_v,\quad A=e(g,v)^a,\quad$$
$$\forall i \in [1..l]: U_i=u_iR_{u,i},\quad  H_i=h_iR_{h,i},\quad  W_i=w_iR_{w,i})$$

{\bf Encryption} uses $PK$ and takes as parameters index attribute $I$ and message $M \in \mathbb{G}_T$. The following random elements are generated: \(Z, Z_{i,1}, Z_{i,2} \in \mathbb{G}_q\) and \(s \in \mathbb{Z}_n\). Then, the ciphertext is: 
$$C = (C^{'}= MA^s,\quad C_0=V^sZ, \quad $$
$$\forall i \in [1..l]: C_{i,1} = (U^{I_i}_iH_i)^sZ_{i,1}, \quad C_{i,2} = W^{s}_iZ_{i,2} )$$

{\bf Token Generation.} Using $SK$, and given a search predicate encoded as pattern vector \(I_{*}\), 
a search token $TK$ is generated as follows: let \(J\) be the set of all indexes $i$ where \(I_{*}[i] \neq *\).
We randomly generate \(r_{i,1}\) and \(r_{i,2} \in \mathbb{Z}_p, \forall i \in J\). 
Then
$$TK=(I_*, K_0 = g^a\prod_{i \in J}(u^{I_{*}[i]}_ih_i)^{r_{i,1}}w^{r_{i,2}}_i, \quad$$
$$ \forall i \in [1..l]: K_{i,1} = v^{r_i,1},\quad K_{i,2} = v^{r_i,2})$$

{\bf Query} is executed at the service provider, and evaluates if the predicate represented by $TK$ holds for ciphertext $C$. The server attempts to determine the value of \(M\) as 
\begin{equation}
M = C^{'}{/} (e(C_0,K_0) {/} \prod_{i \in J} e(C_{i,1},K_{i,1}) e(C_{i,2},K_{i,2}) \label{eq:query}
\end{equation}
If the index $I$ based on which $C$ was computed satisfies $TK$, then the actual value of \(M\) is returned, otherwise a special number which is not in the valid message domain (denoted by $\bot$) is obtained.

The HVE query, or matching, is the most important operation in a location-based alert system, because it is executed every time an alert occurs, and it requires processing of a large number of ciphertexts. {\em Our goal is to reduce the overhead of matching, and the most direct way to do so is by reducing the number of non-star bits in a token, since the number of expensive bilinear maps is proportional to the count of non-star bits}.

\subsection{System Model}\label{System Model}

The architecture of location-based alert systems is shown in Fig.~\ref{Fig: System model}. There are three types of entities: mobile users, a service provider (SP) and a trusted authority (TA). 

 Mobile users  subscribe to the location-based alert system and periodically submit their encrypted location updates. Users want to be notified when they are in an alert cell, without their privacy being compromised. They the public key (PK) of the HVE cryptosystem to encrypt their locations before sending them to the SP. For example, users A and B on the grid encrypt their indexes 110 and 000, generating two ciphertexts $C_A$ and $C_B$, respectively.

The \textbf{Trusted Authority (TA)} has the secret key (SK) of the HVE cryptosystem. In practice, the TA role could be played by a reputable organization such as a law enforcement agency, or the center for disease control, who issue {\em HVE search tokens} corresponding to alerts. The TA does {\em not} have access to user locations, and is assumed not to collude with the SP. The TA is acting in the interest of the general public, but does not have the infrastructure to run a complex alert system, which is why this service is outsourced to the SP.
One important aspect when generating tokens is to minimize the number of non-star bits in a token, in order to reduce the computational overhead of matching. A common approach is to use binary minimization on the cell identifiers. For example, the two alerted indexes 100 and 000 are combined using binary expression minimization to obtain *00, then, the new index is encrypted using the SK to create a token with two non-star bits, instead of two tokens with three non-star bits each. The overhead is reduced from six sets of bilinear pairings to two.

The SP implements the alert service. It receives encrypted updates from users and tokens from the TA, and performs the matching to decide whether encrypted location $C_i$ of user $i$ falls within alert zone $j$ represented by token $TK_j$. If the {\em matching} outcome is positive, the SP learns that the user is inside the alert zone, and notifies the user. For a matching process to result in a positive outcome, all the token's non-star bits should exactly match the user index. Star bits ('don't care' bits), as the name suggests, match with either a zero or one bit in the user index. Note that all received information from users and the TA is encrypted in the matching process, and the search happens over encrypted data only. 

Revisiting the example in Fig.~\ref{Fig: System model}, the outcome of matching between token *00 and user B's ciphertext corresponding to index 000 is positive (all the non-star bits match); however, the matching outcome between *00 and 110 (user A) is negative as the second bits do not match. From the mathematical derivation of HVE (\ref{sec:app}), the HVE system's computation complexity is proportional to the number of non-star bits in the tokens. Therefore, a good grid encoding reduces the overall number of non-star bits in tokens to minimize the HVE computational overhead. 

\newcommand{\rvec}{\mathrm {\mathbf {r}}} 
\begingroup
\begin{table}
\caption {Summary of notations.} 
\vspace{-10pt}
\centering
\begin{tabular}{>{\arraybackslash}m{2cm} >{\arraybackslash}m{5.8cm} }
\hline\hline
 Symbol  & Description\\    \hline
  $n$ & Number of cells \\
  $\mathcal{V} = \{ \bigcup v_i\}$ & Set of all cells\\
  $p(v_i)$ & {Probability of cell $v_i$ becoming alerted}\\
  $C_j$ & Encrypted location of user $j$\\
  $TK_j$ & Token $j$\\
  $M_j$ & Message of user $j$\\
  RL    & Depth of prefix tree (reference length)\\
  $r_i$ & $i$th internal node of tree\\
    Pois($\lambda$)& Poisson distribution; occurrence rate $\lambda$\\
     $\Sigma$ & Identifier symbol alphabet\\
    $\gamma$ &    {\em Euler-Mascheroni} constant\\
    $\phi$  &  Golden ratio  \\
    $a[i:j]$ & Returns elements $i$ to $j-1$ of array $a$\\ 
    $\overline{x_1x_2...x_l}$ & Concatenation of symbols $x_1$ to $x_l$\\
\hline\hline
\end{tabular}
\label{tab:table1}
\vspace{-15pt}
\end{table}
\endgroup

\subsection{Motivation and Scope}

While prior work made important steps toward secure and scalable location-based alert systems, important performance issues still need to be addressed. The pioneering work in~\cite{ghinita2014efficient} was the first to use searchable HVE encryption in the context of locations, but assumed that all data domain regions are equally likely to be part of an alert zone. Later in~\cite{shaham2020enhancing}, it was shown that if there are significant differences in likelihood of distinct regions to be part of an alert zone, then performance can be significantly boosted. However, both~\cite{ghinita2014efficient} and~\cite{shaham2020enhancing} use fixed-length encoding, i.e., the same number of bits are used to represent each cell. Hence, their performance overhead depends entirely on their ability to aggregate search tokens. When alert zones consist of a relatively large number of co-located alert cells, fixed-length encoding methods are able to perform effectively binary minimization of identifiers, and reduce overhead. This may be sufficient in some scenarios such as an active shooter, or a gas leak, where there is an epicenter of the event, and a range around the epicenter (often circular) within which users must be alerted. The range can be large, for instance in the order of hundreds of meters.

However, in other applications, alert zones may be compact and sparse. For instance, consider the case of contact tracing -- an important task in controlling pandemics, such as COVID-19. In this case, there will be a number of distinct alert zones, corresponding to the set of locations visited by a COVID-19 patient. For each individual site, the range of the query is relatively small, for instance, several meters around the patient location for direct spread. Or, in the case of surface spread or aerosol transmission, the query may be restricted to a room, or a store, which may be in the order of $10-20$ meters in size. There are insufficient cells in the alert zones to allow for effective token aggregation with fixed encoding, and the performance obtained may be poor. 

Our goal is to address this latter case, and we do so by using a novel variable-length encoding approach. In this case, it is important to use fewer representation bits for high-probability regions. While our advantage is greatest for small, sparse alert zones, we show in our empirical evaluation in Section~\ref{Experimental Evaluation} that variable-length encoding can outperform fixed-length approaches for a wide choice of alert zone sizes, and mixed-size workloads.

Normalizing the cell probability values over the domain space reveals how likely a cell is to be alerted compared to others. A typical stochastic distribution used to model sporadic events is {\em Poisson distribution}, characterized as follows. 

\begin{thm}\label{Theorem: Poisson distribution}
If a random variable $Y$ represents the number of alert cells on the grid, then, it approximately follows Poisson distribution (Pois($\lambda$)) with the occurrence rate of one ($\lambda = 1$).  
\end{thm}

\begin{proof}
An alert zone event on the map is a subset of cells $v_1,\,,v_2,..., v_{n}$, where $n$ is a large value and each probability $p(v_i)$ is relatively small. Moreover, the events are either independent or weakly dependent of each other. Let 
\begin{equation}
    Y = \sum_{i=1}^n I(v_i)
\end{equation}
count how many of the cells are alerted, in which $I$ is an indicator random variable having a value of one when the cell is alerted and zero otherwise. Based on the Poisson distribution, the random variable $Y$ can be approximated with rate $\lambda = \sum_{i=1}^n p(v_i) = 1$. Therefore, the probability of having $k$ alert cells is given by
\begin{equation}
    p(Y = k)  = \dfrac{e^{-1}}{k!}.
\end{equation}
\end{proof}


One can see from the Poisson distribution that the likelihood of having a large number of alert cells is low. The maximum probability corresponds to having only a single alert cell in a zone, and then it drops significantly. This motivates our technique for dealing effectively with cases where alert zones are compact. 




\section{Location-based Alerts with Variable-length Encoding}

In Section~\ref{sec:huff} we provide an overview of Huffman codes; Section~\ref{sec:encoding} presents the proposed location encoding scheme; Section~\ref{sec:minimize} introduces the token minimization process.

\subsection{Prefix and Huffman Codes}\label{sec:huff}

Generally, any uniquely-decodable representation used to transmit information is a {\em prefix code}, i.e., it follows the {\em prefix property}, which requires that no whole code can be part of any other code. For example, $[000,001,01,10,11]$ is a prefix code as no code starts with any other code in the set. A well-known theorem based on {\em Kraft inequality}~\cite{cover1999elements} states that any prefix over an alphabet of size two with string lengths of $l_1$ to $l_n$ must satisfy the inequality
\begin{equation}
\sum_{i=1}^{n}\dfrac{1}{2^{l_i}} - 1 \leq 0,
\end{equation}
and conversely, given a set of string lengths that satisfies the Kraft inequality, there exists a prefix code with these string lengths. Let the tuple $\mathcal{P} = (p_1,p_2,...,p_n)$ defined over space partitioning $\mathcal{V}$ indicate the likelihood of cells $v_1, \cdots v_n$ becoming alert cells. Furthermore, suppose that the function $f(l_1,l_2,...,l_n)$ returns the average symbol length with no minimization, and $f_M(l_1,l_2,...,l_n)$ returns the average reduction in number of bits in the minimization process. Given the tuple of cells and probabilities, the objective of a minimal encoding is to generate a prefix code $\mathcal{C}(\mathcal{P}) = (c_1,c_2,...,c_n)$ as follows:
\begin{equation*}
\begin{aligned}
& \text{minimize}
& & L(\mathcal{C}(\mathcal{P}))  = \sum_{i=1}^{n}\, p(v_i)\times length(c_i)\\
& \text{subject to}
& & L(\mathcal{C}(\mathcal{P}))\leq L(\mathcal{T}(\mathcal{P}))\; \text{for any code }\; \mathcal{T}(\mathcal{P})\\
\end{aligned}
\end{equation*}

Note that $f_M$, which indicates the amount of minimization, is not necessarily a function. For example, a previously used minimization approach based on Karnaugh maps~\cite{ghinita2014efficient} does not always result in a unique output. The NP-hardness of the above problem based on fixed-length codes is shown in~\cite{shaham2020enhancing}. 


The most well-known prefix code is the {\em Huffman encoding}, widely used in communication systems as it results in optimal decodable average code length. The main idea behind Huffman codes is that more common symbols are represented with fewer bits compared with the less common symbols. In grid encoding,  it is desirable to encode symbols that have higher probabilities of being in alert zones with fewer bits than the less likely ones. Given the tuple of cells and probabilities, the objective of Huffman encoding is to generate a prefix code that minimizes the average length of codewords:

\begin{equation}
\begin{aligned}\label{eq: optimization}
& \text{minimize}
& & f(l_1,l_2,...,l_n) - f_M(l_1,l_2,...,l_n)\\
& \text{subject to}
& & \sum_{i=1}^{n}\dfrac{1}{2^{l_i}} - 1 \leq 0\\
& & & l_i > 0,\;\; \forall i=1,..,n
\end{aligned}
\end{equation}

{\bf Prefix Trees.} An intuitive way to discover whether the prefix property holds for a code is to draw its associated binary tree, called prefix tree. The prefix tree is constructed by assigning an empty character to the root and descending through the tree. At each branching point, we either choose to go left by adding a zero character or move to the right child by adding a character '1' to the root string. We call the tree's depth {\em reference length (RL)}. This number also indicates the maximum length of a prefix code. Moreover, the subtree roots are referred to as {\em interior nodes} of the prefix tree, and the leaf nodes are the {\em prefix codes}. Fig.~\ref{fig:F12}, shows a prefix tree with an RL of three. As an example, the prefix code '001' is generated by traversing nodes $r_4$, $r_2$, and $r_1$.

\subsection{Proposed Coding Scheme} \label{sec:encoding}


The focus of prior work on secure alert zones~\cite{ghinita2014efficient,shaham2020enhancing} has been on fixed-length codes. Such codes are indeed a special case of prefix codes, in which the tree is balanced, and no assigned code can start with another. Next, we show how variable-length codes can be used in conjunction with HVE. An overview of the proposed approach is presented in Fig.~\ref{Fig: Overview}. Based on a given prefix code, the TA generates grid indexes where each index is a unique identifier of a cell in the grid. In addition to grid indexes, a coding tree is generated for the purpose of token minimization. Given the set of indexes associated with the alert cells, the TA applies the proposed minimization algorithm and transmits the encrypted tokens to the SP. Fig.~\ref{fig: algorithm process} serves as a running example.

\begin{figure}[t]
\centering
\includegraphics[scale=.5]{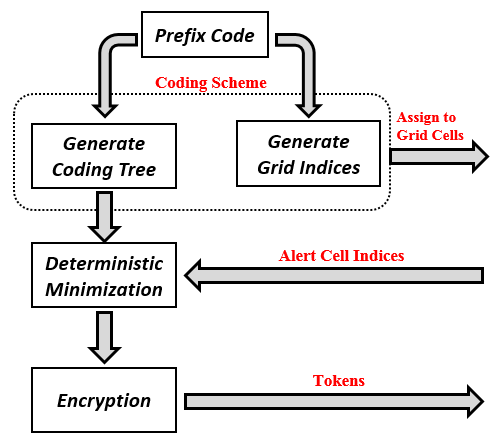}
\hspace{1em}
\centering
\caption{Overview of HVE with variable-length codes.}
\label{Fig: Overview}
\vspace{-15pt}
\end{figure}

Our approach consists of four steps:

{\em I. Generation of Probabilities:} Our coding scheme relies on a set of probabilities for each cell of the location domain to be part of an alert zone. This step is a prerequisite to our approach, and thus performed independently of the encoding. Such probabilities are application dependent, and can be generated based on a trained machine learning model. In the example of Fig.~\ref{fig:F11}, we have five cells $\mathcal{V}=( v_1,\,v_2,\, v_3\, v_4,\, v_{5} )$ with alert probabilities of $$\mathcal{P}\Equal( p(v_1)\Equal 0.1,\,p(v_2)\Equal 0.2,\,p(v_3)\Equal 0.5,\,p(v_4)\Equal 0.4,\, p(v_{5})\Equal 0.6 ).$$

For grids entailing a high correlation between alert probabilities of cells, the setting in~\cite{shaham2020enhancing} or deep learning models such as~\cite{almanie2015crime} can be used to find the stationary distribution of probabilities, leading to a more accurate probabilistic model.

{\em II. Prefix tree :} An arbitrary prefix code defined over alphabet $\Sigma = \{0,1\}$ can be represented by a binary tree with the prefix codes located on the leaves of the tree. We are not just interested in the generated prefix codes, but also in the codes associated with the internal nodes of the tree. Therefore, internal nodes are also stored as well as the generated prefix codes. 

The topology of the tree is stored by recording five attributes of each node: left child (leftChild), right child (rightChild), parent node (parentNode), weight, and the associated code. The weight of a node represents its frequency. The leaf nodes have a frequency equal to their probability, and the weight of a parent node is found by the addition of its immediate children's weights (i.e., {\em Huffman mechanism}). The prefix tree is not used directly in the prefix coding scheme, but two sets of codes are generated based on the prefix tree; one is used for identifying grid cells referred to as cell {\em indexes}, and another is used by the TA to perform token minimization. 
Once the base codes are assigned for each node of the tree, two sets of {\em padding} are conducted, one for indexes assigned to the cells, and one used as a guideline for the token generation. The padding leads to a length of RL (i.e., equal length) for all codewords and indexes. Recall that equal ciphertext lengths is a requirement for security. However, the variable-length codes affect the ciphetexts and token contents in a way that allows fast processing. Furthermore, the padding prevents an adversary from distinguishing among ciphertexts.\\ 

\begin{algorithm}[t]
\caption{Coding Scheme}\label{Algo: Coding Scheme}
\begin{flushleft}
    \hspace*{\algorithmicindent} \textbf{Input}:\;\;  Root; $\mathcal{V}$;
\end{flushleft}
\begin{algorithmic}[1]
\State //Root traversal to generate codes 
\Function{Traverse}{{\em Root}} 
    \If {{\em Root} has no children}
    \State \textbf{return} True
    \Else
    \State {\em Root}.leftChild = {\em Root}.code + '0'
    \State {\em Root}.rightChild = {\em Root}.code + '1'
    \State Traverse(leftChild)
    \State Traverse(rightChild)
    \EndIf
\EndFunction\\

\State $Traverse\,(Root)$
\State //Generate indexes assigned to cells
\State  $RL \leftarrow \text{depth of tree}$
\For {all leaf nodes} 
    \State $index = node$.code
    \While{$len(node.\text{code})<RL$}
        \State $index = index + '0'$
    \EndWhile
    \State Assign index to $v_i$ that has $p(v_i)= node.\text{weight}$
\EndFor\\

\State //Generate coding tree
\For {all nodes} 
    \While{$len(node.$code$)<RL$}
    \State $node$.code = $node$.code + '*'
    \EndWhile
\EndFor
\State //{\em codingTree} is the set of all nodes, alternatively {\em Root} can be returned
\State \textbf{return} {\em codingTree}
\end{algorithmic}
\end{algorithm}

III. {\em Grid indexes:} the prefix codes (leaves on the prefix tree) are padded from the right-hand side with zeros if they have a length less than RL. In our example, the generated prefix codes are $\{v_1:001, v_2:000, v_3:10, v_4:01, v_5:11\}$ which are transformed to  $\{v_1:001, v_2:000, v_3:100, v_4:010, v_5:110\}$ after padding with zeros. We refer to zero-padded prefix codes as {\em indexes}. Once codes are created, they are assigned to corresponding cells identified by their probabilities. The assigned indexes to the sample grid are shown in Fig~\ref{fig:F14}. These are the indexes utilized by users to identify the cell they are enclosed by.\\

IV. {\em Coding tree:} the coding tree is used by the trusted authority to generate tokens. The coding tree is constructed by adding star bits on the right side of the prefix codes as well as the internal nodes on the prefix tree if they have a length less than RL. The padding for the sample grid is shown in Fig.~\ref{fig:F13}. The codes on the coding tree are referred to as {\em codewords}.\\

Algorithm~\ref{Algo: Coding Scheme} formally presents how indexes and the coding tree are generated for a given prefix tree. The inputs to the algorithm are the tree root, grid cells, and their probabilities. The tree root is sufficient for reconstructing the tree as children and parents are presumed to be recorded. The algorithm traverses through nodes to generate the prefix tree. Next, indexes of the grid are generated and assigned to the grid cells, and finally, the coding tree is completed and returned as the output of the algorithm.


Algorithm~\ref{Algo: Huffman Tree} presents how the Huffman tree is generated. The algorithm starts by creating a node (leaf node) for each cell of the grid, sorting them in ascending order based on their weights, and placing them in a priority queue. Recall that the weights of the leaf nodes are the probability of cells becoming alerted. Next, while the length of the queue is greater than one, the algorithm extracts two nodes with the minimum weights and creates a new internal node (newNode) with a weight equal to the addition of two extracted nodes. The new node is assigned as the parent of extracted nodes, and the extracted nodes are assigned as left and right children of the parent node. The new node's weight is inserted in the queue, and the process continues until only a single weight remains in the queue. The last node is the root of the tree and the output of the algorithm. The root node is used as input to Algorithm~\ref{Algo: Coding Scheme} to generate the coding tree and grid indexes. The algorithm is executed with the time complexity of $\mathcal{O}(n(\log_2n))$.

\begin{algorithm}[t]
\caption{Huffman Tree}\label{Algo: Huffman Tree}
\begin{flushleft}
 \hspace*{\algorithmicindent} \textbf{Input}:\;\;   $\mathcal{V}$; $\mathcal{P}$
\end{flushleft}
\begin{algorithmic}[1]
\State //Generate tree nodes
\For {$v_i\in\mathcal{V}$} 
	\State Create a newNode(leftChild=None, rightChild = None, 
	\State \qquad \qquad  \qquad parent = None, weight = $p(v_i)$, code = ' ')

\EndFor

\State Insert nodes into priority queue $Q$
\While{len($Q$)>1}
    \State Sort $Q$ in ascending order of weights
    \State  $(node_1, node_2)\leftarrow$   Extract first two nodes in $Q$
	\State Create a newNode(leftChild= $node_1$, rightChild = $node_2$, 
	\State \qquad \qquad   parent = None, 
	\State \qquad \qquad   weight = $n_1.weight+ n_2.weight$, code = ' ')
    \State $n_1$.parent, $n_2$.parent= newNode
    \State Insert newNode into $Q$
\EndWhile\label{euclidendwhile}
\State //The last nodes in $Q$ is the tree root
\State \textbf{return} root 
\end{algorithmic}
\end{algorithm}

\begin{figure}[t]
\centering
	\subfloat[Sample grid.\label{fig:F11}]{%
	\includegraphics[scale = 0.5]{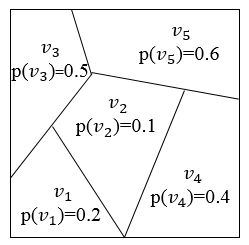}
	}	
	\hfill
	\subfloat[Coding tree generated based on Huffman encoding.\label{fig:F12}]{%
		\includegraphics[scale = 0.5]{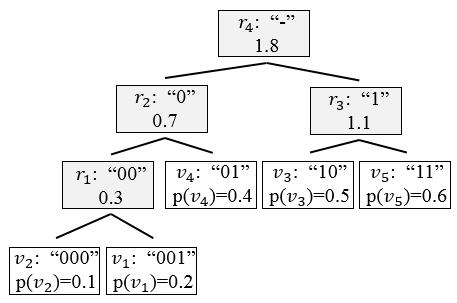}	
	}
	\hfill
	\subfloat[Assigned grid indexes.\label{fig:F14}]{%
		\includegraphics[scale=.5]{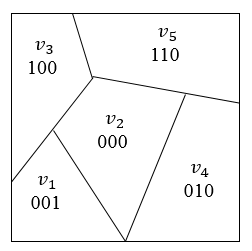}
	}
	\hfill
	\subfloat[Coding tree.\label{fig:F13}]{%
		\includegraphics[scale=.5]{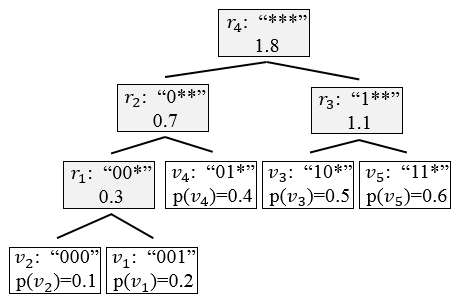}
	}
	\vspace{-5pt}
	\caption{Sample variable-length coding scheme}
	\label{fig: algorithm process}
	\vspace{-15pt}
\end{figure}

The following steps illustrate the generation of Huffman tree for the example presented in Fig.~\ref{fig: algorithm process}.

1. One node is generated for each cell $( v_1,\,v_2,\, v_3,\, v_4,\, v_{5} )$ and their probabilities are inserted in a priority queue $$Q\Equal( p(v_1)\Equal 0.2,\,p(v_2)\Equal 0.1,\,p(v_3)\Equal 0.5,\,p(v_4)\Equal 0.4,\, p(v_{5})\Equal 0.6 ).$$

2. The queue is sorted in an ascending order:
$$Q\Equal( p(v_2)\Equal 0.1,\,p(v_1)\Equal 0.2,\,p(v_4)\Equal 0.4,p(v_3)\Equal 0.5,\, p(v_{5})\Equal 0.6 )$$
    
3. The two nodes with the minimum weights ($v_1$ and $v_2$) are extracted from the queue and a new parent node $r_1$ is generated with the weight of $p(v_2)+p(v_1)=0.3$ and inserted into the queue:
$$Q\Equal( p(r_1)\Equal 0.3,\,p(v_4)\Equal 0.4,\,p(v_3)\Equal 0.5,\, p(v_{5})\Equal 0.6 )$$

4. Similarly $r_2$, $r_3$, and $r_4$ are generated as 
\begin{align}
&Q\Equal( p(r_2)\Equal 0.7,\,p(v_3)\Equal 0.5,\, p(v_{5})\Equal 0.6 ),\nonumber\\
&Q\Equal( p(r_2)\Equal 0.7,\,p(r_3)\Equal 1.1),\nonumber \\
&Q\Equal( p(r_4)\Equal 1.8). \nonumber
\end{align}


Another prefix tree evaluated in the experiments is called balanced tree. This prefix tree is used as a baseline to understand the improvement made by the Huffman tree. The balanced tree is a complete binary tree constructed in $log_2(n)$ steps. Given a tuple of probabilities corresponding to grid cells, they are sorted in ascending order and placed in a priority queue, i.e., $Q$. In the $j$th step, nodes $Q[2i]$ and $Q[2i+1]$ are paired for $i=0,1,..., \left \lfloor{n/2^j}\right\rfloor-1$, and each pair is replaced with a parent node in the queue. The weight of a parent is the addition of its immediate children's weights. The final remaining node in the queue is the tree's root.

\subsection{Token Generation and Minimization}
\label{sec:minimize}

Prior work~\cite{ghinita2014efficient,shaham2020enhancing} showed how the process of token generation for an alert zone can considerably improve the computation overhead, if the process of {\em token aggregation} is performed. Specifically, the binary codes corresponding to different regions of an alert zone can be aggregated to yield tokens with few non-star symbols, which in turn reduces the HVE overhead. Binary minimization on fixed-length codes is used for this purpose. For instance, suppose that the alert zone contains cells $0000,\,0010,\,0110,\,0100,$. Instead of separately encrypting the cell indexes and generating four tokens, the TA uses binary minimization to generate a single token $0**0$, and the cost is reduced from twelve HVE operations to two. 
Binary minimization works when there are many cells in the alert zone, and when the placement of these cells permits code minimization. This approach is suitable when the number of alert cells is significant; however, in practice, alert zones may have cell configurations that do not permit efficient aggregation.

We propose a different token generation approach, where instead of performing binary minimization on fixed-length codes, we control the configuration of tokens based on the assignment of variable-length codes to cells. Algorithm~\ref{Algo: Huffman-HVE Minimization} summarizes this process. Inputs to the algorithm are a set of alert cells and the coding tree. In the initialization phase, the algorithm defines: 

\begin{itemize}
    \item a dictionary of parent nodes ({\em parentDict}) with the number of leaf nodes in the corresponding subtree. This is done by traversing through children of parent nodes and counting the number of leaves located in that subtree. For the sample example, we have the dictionary as  $$[00* : 2,\, 0\!*\!* : 3,\, 1\!*\!* : 2,\, *\!*\!*: 5 ]$$
    \item a list of leaf nodes denoted by {\em leaves}, ordered as they appear on the tree while traversing; no two edges of the tree cross path. Such a list for the sample tree is: $$[v_2:000, v_1:001, v_4:01*,  v_3:10*, v_5:11*].$$
\end{itemize}

The algorithm continues by converting alert cell indexes to codewords on the tree and recoding their associated codeword and the corresponding index in  {\em leaves}. By default, the mapping process splits codewords into clusters that are located consecutively in {\em leaves}. It is important to note that mapping of alert cell indexes to codewords is unique, as demonstrated in Theorem~\ref{Thm: bijective mapping}. The theorem proves a bijective mapping between grid indexes and coding tree codewords. For instance, if the alert cells are $[001,100,110]$, then the mapping would result in leaves $[001,10*,11*]$ for the sample example. Next, the minimization process based on the coding scheme is conducted. The minimization's main idea is to find the common subtree roots that have maximum depths and use them as tokens. All leaves under a common subtree root must be alerted; otherwise, if a user is located in such a leaf node it will be falsely notified to be in an alert zone. 

Continuing with the example and alert cells $[001,10*,11*]$, the algorithm generates two clusters $[10*,11*]$ and $[001]$, and aims to identify the common subtree roots with the maximum depths in each cluster. This is done heuristically in lines~\ref{Start while loop}-~\ref{End while loop}. Suppose that a cluster's length is $L$, the common left-hand side code in all $L$ codewords is calculated and padded with '*' bits to ensure that the codeword length is RL. If the common codeword exists in the dictionary and the number of its children is $L$, the codeword is chosen as representative of its descendent leaves; otherwise, $L$ is decremented by one, and now the first $L-1$ members are checked to see if there exists a common root associated with them. The process continues until the first subtree root is found. For the remaining codewords in the cluster, the algorithm is applied again until all tokens representing codewords in the cluster are selected. A similar approach is repeated for all clusters.


\begin{algorithm}[t]
\caption{Deterministic Minimization}\label{Algo: Huffman-HVE Minimization}

\begin{flushleft}
\hspace*{\algorithmicindent} \textbf{Input}:\;\; {\em alertCells}; {\em codingTree};
\end{flushleft}
\begin{algorithmic}[1]
\State {\em parentDict} = $\{ \}$
\For {{\em node} $\in$ {\em codingTree}} 
    \State {\em parentDict}[{\em node}.code] = \# descendent leaves
\EndFor

\State  {\em indexHolder, codewordHolder} = []
\State {\em leaves} $\leftarrow$ list of leaf codewords
\For{$i \in$ {\em alertCells}}
    \State {\em memCodeword} $\leftarrow$ Map $i$ to a codeword in {\em leaves}
    \State  {\em codewordHolder} =  {\em codewordHolder} $\cup \{ memCodeword\}$
    \State {\em memIndex} $\leftarrow$ index of {\em memCodeword} in {\em leaves} 
    \State  {\em indexHolder} =  {\em indexHolder} $\cup \{ memIndex\}$
\EndFor

\State // Generate a two dimensional list of clusters
\State {\em Clusters, c} = []
\State $c = c \cup codewordHolder[0]$
\For {$i\in [1:len(codewordHolder)]$}
    \If{$indexHolder[i] = indexHolder[i-1]+1  $}
       \State $c = c \cup codewordHolder[i]$
     \Else 
        \State {\em clusters} $ = $ {\em clusters} $\cup\, c$
        \State $c$ = []
        \State $c = c \cup codewordHolder[i]$   
    \EndIf
\EndFor

\State {\em tokens} = []
\State  $RL \leftarrow \text{depth of tree}$
\For { {\em cluster} $\in$ {\em clusters}} \label{Start while loop} 
    \State $L = len(cluster)$
    \While{$L>1$}
        \State $code \leftarrow$ common bits in $cluster[1:L]$
        \If {$len(code)< RL $}
            \State Pad with $ RL-len(code)$ star bits
        \EndIf        
        \If {$code \in parentDict\,\, \&\,\, parentDict[code] = L$}
            \State $tokens = tokens \cup code$ 
            \State $cluster = cluster[L:len(cluster)]$
            \State $L = len(cluster)$
        \Else 
            \State $L = L-1$
            \If {$L=1$}
                \State $tokens = tokens \cup cluster[L]$   
                \State $cluster = cluster[L:len(cluster)]$
            \EndIf       
        \EndIf
    \EndWhile\label{End while loop}
\EndFor

\State \textbf{return} {\em tokens}
\end{algorithmic}
\end{algorithm}

\begin{thm}\label{Thm: bijective mapping}
    There exists a bijective function between grid indexes and the leaf nodes of the coding tree.
\end{thm}

\begin{proof}
    We start by proving that for each index on the grid there exists a unique leaf node (codeword) on the tree. Let $\overline{x_1x_2...x_l}$ denote an arbitrary index on the map. There exists at least one leaf on the tree with the codeword $\overline{y_1y_2...y_{r_1}*...*}$ such that $\overline{x_1x_2...x_{r_1}} = \overline{y_1y_2...y_{r_1}}$, as indexes have been generated from leaf nodes of the prefix tree. Suppose that there exist at least two leaf nodes with the codewords $\overline{y_1y_2...y_{r_1}*...*}$ and $\overline{z_1z_2...z_{r2}*...*}$ corresponding to the index $\overline{x_1x_2...x_l}$. Hence, we have the following relationship between the index and codewords on the tree. 
    \begin{align}
        \overline{x_1x_2...x_{r_1}} = \overline{y_1y_2...y_{r_1}}\label{Equ: d1}\\
        \overline{x_1x_2...x_{r_2}} = \overline{z_1z_2...z_{r_2}}\label{Equ: d2}
    \end{align}
        
    Without loss of generality, assume that $r_2\geq r_1$. Hence, equations \ref{Equ: d1} and \ref{Equ: d2} result in
    \begin{equation}
        \overline{y_1y_2...y_{r_1}} = \overline{z_1z_2...z_{r_1}}.
    \end{equation}
    However, this contradicts the prefix property of the codes. Hence, there is a unique leaf node corresponding to each cell index. As there are an equal number of indexes and codewords, there exists a bijective mapping between indexes and codewords.  
\end{proof}

\section{Extension to Non-Binary Codes}


So far, we considered the alphabet of HVE operations to be limited to $\Sigma = \{0,\,1\}$ and the extended alphabet as $\Sigma_* = \Sigma\cup\{*\}$. This is an intuitive way of looking at indexes as they are a series of zeros and ones. However, by extending the alphabet to $\Sigma = \{0,\,1,...,\,B-1\}$ for an arbitrary integer $B\in \{2,..., n-1\}$, we could obtain more compact representations. The special character is also added as $\Sigma_* = \Sigma\cup\{*\}$. 
We re-visit the operations from the previous section for the extended alphabet with $B$ symbols.

1. {\em Prefix tree:} We incorporate an extension of Huffman coding referred to as $B$-ary Huffman to generate the prefix tree. 
The main idea is to group $B$ least probable symbols (instead of $2$) at each substitution stage of the algorithm. The construction of the prefix tree for our running example grid is shown in Fig.~\ref{Fig: Extension1} in which a $3$-ary or Huffman code is used. Initially, the algorithm starts by combining nodes $v_2$, $v_1$, and $v_4$, as they correspond to a group of three nodes with the minimum total weight, generating the node $r_1$. Next, the nodes $r_1$, $v_3$, and $v_5$ are combined, and the root node $r_2$ is generated. The weights and other characteristics of the nodes are stored and calculated in the same way as the binary Huffman tree. The codes associated with the tree are generated by assigning an empty character to the root node and then traversing the tree. At each branching node, when following the $i$th child edge, character $i-1$ is added to the root string. As an example, prefix code '02' is generated by adding character '0' at $r_1$, and character '2' by moving to node $v_4$. As in the case of the binary case, we are interested in codes assigned to internal nodes as well as the prefix codes generated at the leaves.  


\begin{figure}[t]
	\subfloat[]{%
	\includegraphics[scale=.45]{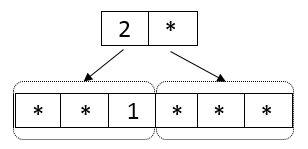}
	\label{Fig: holder1}
	}
	\hfill
	\subfloat[]{%
	\includegraphics[scale=.45]{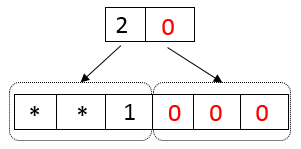}
	\label{Fig: holder2}
	}
	\vspace{-10pt}
	\caption{Expansion process.}
	\vspace{-25pt}
\end{figure}

2. {\em Coding tree:} The generation of the coding tree requires an additional step compared to the binary Huffman tree.  In the first step, codes are padded with star characters until they reach the same length as the RL. The padded prefix tree for our running example is shown in Fig.~\ref{Fig: Extension2}. Next, we expand each character to an array of $B$ bits. The character $i\in \Sigma$ is converted to $B$ bits with the ($i$+$1$)-th bit set to 1 and star bits otherwise. The only exception is the star character, which will be mapped to a string of length $B$ with all bits set to '*'. As an example, the expansion of $2*$ is shown in Fig.~\ref{Fig: holder1}.

Each original character essentially works as a placeholder for the expanded representation.  The final coding tree generated for our example is shown in Fig.~\ref{Fig: Extension3}.

\begin{figure*}[b]
\centering
	\subfloat[Ternary Huffman coding tree.\label{Fig: Extension1}]{%
	\includegraphics[scale = 0.45]{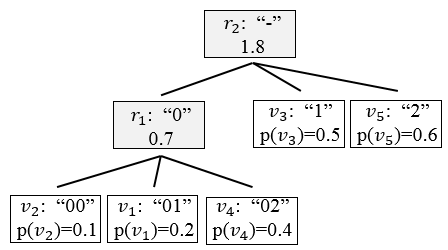}
	}	
	\subfloat[Placeholders.\label{Fig: Extension2}]{%
		\includegraphics[scale = 0.45]{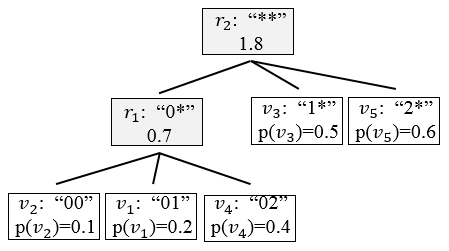}	
	}
	\subfloat[Coding tree.\label{Fig: Extension3}]{%
		\includegraphics[scale=.42]{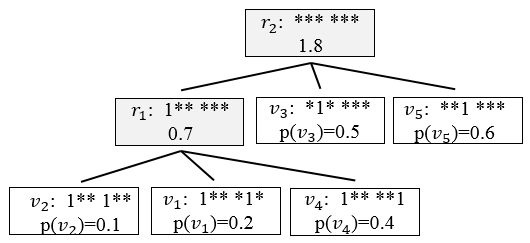}
    }
\vspace{-10pt}
	\caption{Sample coding tree for extended framework.}
\end{figure*}

3. {\em Indexes:} We generate indexes by padding the leaves in the prefix tree by zeros, and then expanding the codes. An interesting case occurs that gives the TA the opportunity to increase the grid's granularity further if desired. Consider the prefix code '2', which will be zero-padded to generate '20'. The expansion process requires two steps: (i) zeros generated by the padding process are mapped to $B$ bits; (ii) each character $i\in \Sigma$ is expanded to $B$ bits with the ($i$+$1$)-th bit set to 1, and star bits otherwise. The expansion of '20' is demonstrated in Fig.~\ref{Fig: holder2}. 


The additional star bits in the index are converted to zeros. The advantage of the approach is revealed when we increase the granularity of a grid cell in a later stage in time. This can be done by exploiting the star bits generated in the last step without violating the structure of the grid or the coding tree. Consider the index '20' corresponding to cell $v_5$ one more time. This string was first converted to '**1000' and then to '001000'. Suppose, later on, the TA decides to increase the granularity of $v_5$ to four cells. This can simply be done by using four indexes '001000', '011000', '101000', '111000' generated based on star bits with all of them lying under character '2'. The coding tree is also updated accordingly via placeholders for the character '2' without violating the tree's prefix property.

\section{Encryption Overhead}

Employing variable-length codes into HVE can significantly improve the computation complexity at the SP, but there is a a trade-off with respect to increased encryption time. When variable-length encoding is used, all ciphertexts submitted by the mobile users to the SP must have the maximum length of any existing code. Otherwise, the length of the ciphertext would enable the SP to pinpoint the location of the submitting user to one of the cells that are assigned a code with bit length equal to the one submitted. To thwart such attacks, all codes are padded before encryption to the maximum possible length, i.e. RL. In this section, we analyze this additional encryption overhead, and we show that it is not significant, especially compared to the savings at the SP. Furthermore, the additional computational load is spread over the user population, since each user encrypts its own location independently, and no bottleneck is created (as opposed to the alert matching overhead which is centrally incurred at the SP).

In our analysis, we make use of the following result:
\begin{thm}\label{Thm: Huffman Length}
    The depth of a B-ary Huffman tree (RL) with $n$ leaves is less than or equal to $\lceil \dfrac{n-1}{B-1}  \rceil$.
\end{thm}
\begin{proof}

The theorem can be proved by counting the number of internal nodes in a B-ary Huffman tree. Consider a tree with $n$ leaves generated by the Huffman mechanism. At every run of the algorithm, $B$ less likely remaining nodes in the priority queue are combined, and a new internal node is inserted. Suppose that the Huffman mechanism is conducted $x$ times over the priority queue until a single node, i.e. root node, is left in the queue. The maximum value of integer $x$ can be derived as:

\begin{equation}
\underset{x}{max}\;\;\{ n - x(B-1)\geq 1\} \rightarrow x =  \lceil \dfrac{n-1}{B-1}  \rceil  
\end{equation}
Therefore, the maximum possible depth of a B-ary Huffman tree is $\lceil \dfrac{n-1}{B-1}  \rceil$.
\end{proof}

Let $L_E$ denote the difference between the RL of an encoding grid with $n$ cells generated by Huffman coding and fixed-length codes. We start by deriving an upper bound for $L_E$ when $\Sigma_* = \{0,1 \} \cup\{*\}$, and then extend the upper bound for an arbitrary size alphabet. Without loss of generality, consider that RL in the binary Huffman tree is $l_n$. The minimum possible value for $l_n$ is $\lceil \log_2n \rceil$. Based on Theorem~\ref{Thm: Huffman Length}, $L_E$ can be written as:
\begin{equation}
    L_E(B=2,n) = l_n - \lceil \log_2n \rceil \leq \lceil \dfrac{n-1}{2-1}  \rceil - \lceil \log_2n \rceil = n-1 -\lceil \log_2n \rceil
\end{equation}

A tighter upper-bound for RL in a binary Huffman tree can be derived based on the following theorem proven in~\cite{buro1993maximum} (we omit the proof for brevity):  

\begin{thm}\label{Thm: tight Huffman upperbound}
    Let $p_n$ and $l_n$ denote the minimum probability and its corresponding length existing on the Huffman tree. Then, 
    \begin{equation}
        l_n\leq log_\phi{\dfrac{1}{p_n}}
    \end{equation}
    where $\phi$ denotes the {\em golden ratio}, i.e., $\phi = (1+\sqrt{5})/2$.
\end{thm}
\noindent Therefore, a tighter upper-bound for $L_E$ can be written as 
\begin{equation}
    L_E(B=2,n) \leq  log_\phi{\dfrac{1}{p_n}} -\lceil \log_2n \rceil
\end{equation}

\begin{figure}[t]
\centering
\includegraphics[scale=.55]{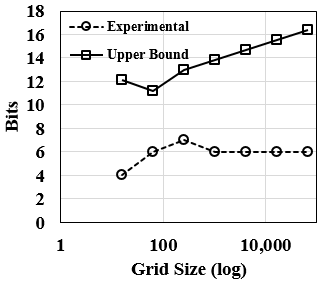}
\hspace{1em}
\centering
\vspace{-10pt}
\caption{ Upper bound of $L_E$ for Binary Huffman codes.}
\vspace{-15pt}
\label{Fig: Upper Bound}
\end{figure}

The numerical and analytical values of $L_E$ are verified\footnote{Grid probabilities are generated with the parameters of sigmoid function set to $a = 0.95$ and $b=20$. Please refer to Section 6 for details.} for binary Huffman coding in Fig.~\ref{Fig: Upper Bound}. 

Now let us extend the approach for B-ary Huffman codes generated with the alphabet $\Sigma_* = \{0,\,1,...,\,B-1\}\cup\{*\}$. Based on information theory, the minimum length of RL corresponding to fixed-length codes is derived as $\lceil \log_Bn \rceil$. Therefore, the upper-bound for $L_E$ can be computed as: 

\begin{align}
    L_E(B,n) = B(l_n - \lceil \log_Bn \rceil) &\leq B(\lceil \dfrac{n-1}{B-1}  \rceil - \lceil\log_Bn\rceil)\\
                           &\leq B( \dfrac{n-1}{B-1} +1  - \lceil\log_Bn\rceil)
\end{align}
The multiplier $B$ is required to map the alphabet to $0$s and $1$s, used in the encryption.

\begin{align}
    E[L_E(n)] \leq &\dfrac{1}{n-1}(\mathlarger{\sum}_{i= 2}^{n}\dfrac{i(n-1)}{i-1}
                   +\mathlarger{\sum}_{i= 2}^{n} i
                   - \mathlarger{\sum}_{i= 2}^{n} i\lceil\log_in\rceil)
\end{align}

The first and second summation in the upper-bound of $E[L_E(n)]$ can be further simplified as
\begin{align}
    \mathlarger{\sum}_{i= 2}^{n}\dfrac{i(n-1)}{i-1} = &(n-1)\times (n-1+\mathlarger{\sum}_{i= 2}^{n}\dfrac{1}{i-1})\\
    & \approx (n-1)\times (n-2 + \ln (n-1) +\dfrac{1}{2(n-1)}+ \gamma)
\end{align}
and,
\begin{equation}
    \mathlarger{\sum}_{i= 2}^{n} i= \dfrac{n^2+n-1}{2}
\end{equation}

where $\gamma\approx 0.577$ is the {\em Euler-Mascheroni} constant. The approximation for the $n$th Harmonic can be derived by its asymptotic expansion in the {\em Hurwitz zeta} function~\cite{chen2015ramanujan}.


\section{Security Discussion}\label{sec:discussion}

Our proposed technique uses as building block HVE primitives as introduced in~\cite{boneh2007conjunctive}, and hence inherits the security properties of HVE, namely IND-CCA under the bilinear Diffie-Hellman assumption. In terms of ciphertext processing semantics, the security achieved by our technique is similar to existing work in the area of secure computation, namely the only leakage that occurs as part of ciphertext matching is the evaluation outcome. Specifically, the SP learns only whether the user is included in the alert zone (which is a necessary condition for correctness), and no other information. The SP does not learn where exactly the user is located within the alert zone, if the match is successful; conversely, if the match is not successful, the SP  learns only that the user is not inside the alert zone, but cannot further narrow down the user within the data domain.

Furthermore, our technique is guided by statistical information that is derived solely from public data. Namely, the heuristic on how to encode cells does not use any user location data, but strictly likelihood scores that are assigned to grid cells, based on public knowledge regarding the {\em alert zone} properties, such as site popularity, etc. No private information regarding any system user is included in the encoding process (not even aggregate data, such as user distribution, etc).

Finally, the encryption strength achieved by HVE depends on the underlying bilinear pairing curve used~\cite{boneh2007conjunctive}. Modern elliptic-curve pairing-based cryptography can easily provide 128-bit security, which is on par with commercial database applications such as banking, or heathcare data security standards.

\section{Experimental Evaluation}\label{Experimental Evaluation}

We conduct our experiments on a $3.40$GHz core-i7 Intel processor with 8GB RAM running $64$-bit Windows $7$ OS. The code is implemented in Python.
We evaluate our methods on both real and synthetic datasets, as follows:

\begin{itemize}
    \item {\em Chicago Crime Dataset.} This dataset is provided by the Chicago Police Department's CLEAR (Citizen Law Enforcement Analysis and Reporting) system~\cite{chicago}. The dataset consists of reported incidents of crime that occurred in the city of Chicago in 2015. We consider four categories of crime: homicide, sexual assault, sex offense, and kidnapping. Fig.~\ref{Fig: chicago stats} shows data statistics. A $32\times32$ grid is overlaid on top of the dataset, and a logistic regression model is trained with the crime data from January to November 2015, and tested on the December data. The accuracy of the model is $92.9\%$ and the generated likelihood scores based on the model are used as input to our techniques. 
    \item {\em Synthetic data.} We generate the likelihood of grid cells to be part of an alert zone using a sigmoid activation function $\mathcal{S}(X=x)= 1/(1+\exp^{-b(x-a)})$, where $a$ and $b$ are parameters controlling the function shape. For each data point (i.e., cell) $x$, a uniformly random number between zero and one is generated, i.e., $x\in X \sim\text{uniform}(0,1)$. Then, the number is fed into the sigmoid activation function. The output is a value between zero and one indicating the likelihood of the cell to be inside an alert zone. The sigmoid function is a frequent model used in machine learning, and we choose it because we expect that, in practice, the probability of individual cells becoming part of an alert zone can be computed using such a model built on a regions' map of features (e.g., type of terrain, building designation, etc.). Parameter $a$ of the sigmoid controls the \textit{inflection} point of the curve, whereas $b$ controls the gradient.
\end{itemize}


\begin{figure}[t]
\centering
\includegraphics[scale=.55]{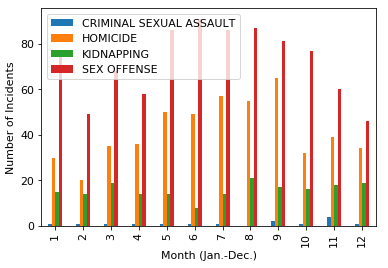}
\hspace{1em}
\centering
\vspace{-10pt}
\caption{Chicago crime dataset statistics.}
\label{Fig: chicago stats}
\vspace{-20pt}
\end{figure}

\begin{figure}[t]
	\subfloat[]{%
	\includegraphics[scale=.29]{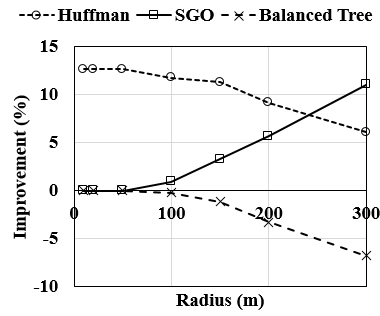}
	}
	\hfill
	\subfloat[]{%
	\includegraphics[scale=.29]{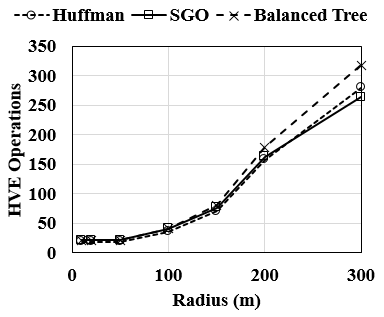}
	}
	\vspace{-10pt}
	\caption{Evaluation on Chicago crime dataset.}
	\label{Fig: Chicago performance evaluation}
	\vspace{-20pt}
\end{figure}

\begin{figure*}[h!]
	\subfloat[a=0.9, b=10 ]{%
	\includegraphics[scale=.29]{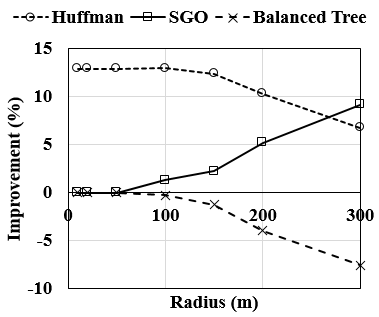}
	}
	\hfill
	\subfloat[a=0.9, b=10 ]{%
	\includegraphics[scale=.29]{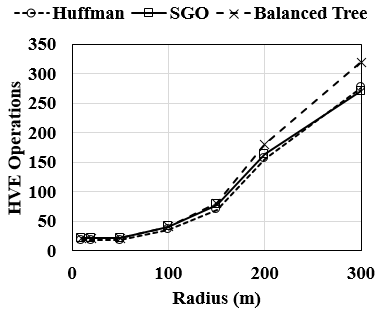}
	}
	\hfill
	\subfloat[a=0.9, b=100 ]{%
	\includegraphics[scale=.29]{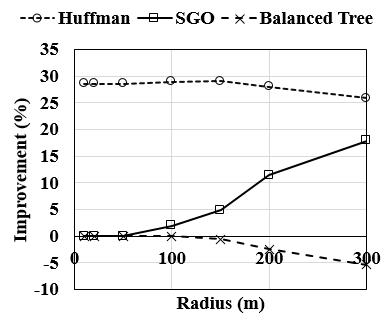}
	}
	\hfill
	\subfloat[a=0.9, b=100 ]{%
	\includegraphics[scale=.29]{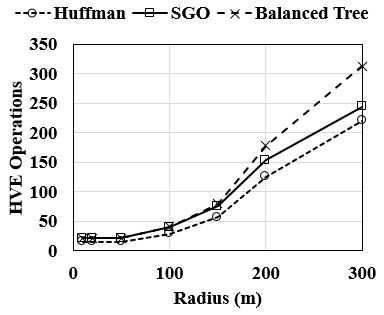}
	}
	\vspace{-10pt}
	\hfill
	\subfloat[a=0.9, b=200 ]{%
	\includegraphics[scale=.29]{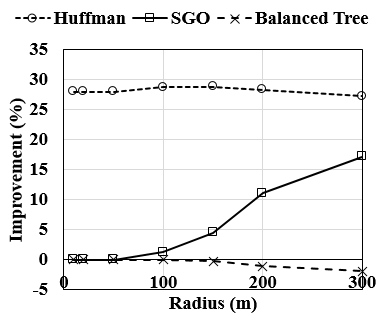}
	}
	\hfill
	\subfloat[a=0.9, b=200 ]{%
	\includegraphics[scale=.29]{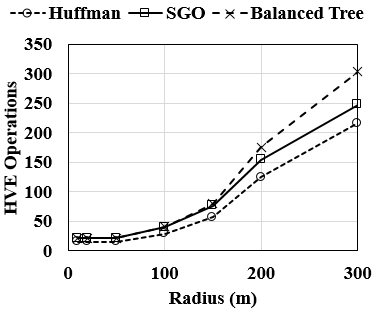}
	}
	\hfill
	\subfloat[a=0.99, b=10 ]{%
	\includegraphics[scale=.29]{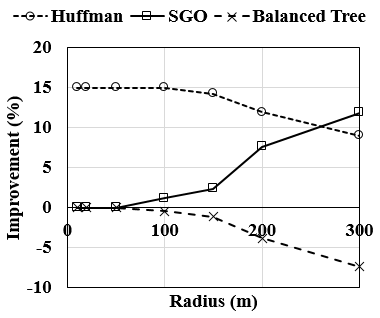}
	}
	\hfill
	\subfloat[a=0.99, b=10 ]{%
	\includegraphics[scale=.29]{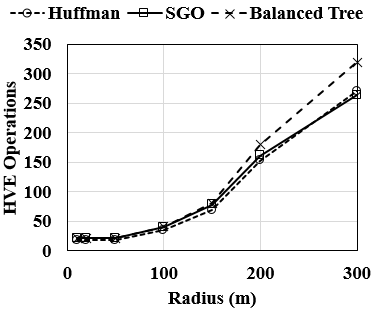}
	}
	\vspace{-10pt}
	\hfill
	\subfloat[a=0.99, b=100 ]{%
	\includegraphics[scale=.29]{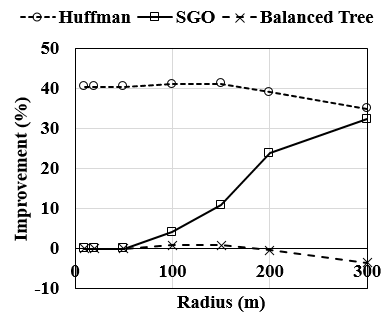}
	}
	\hfill
	\subfloat[a=0.99, b=100 ]{%
	\includegraphics[scale=.29]{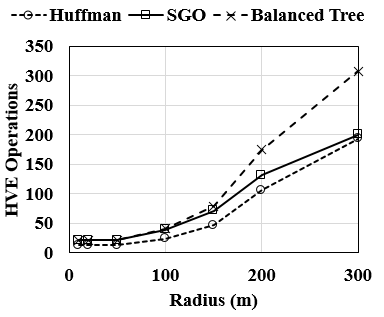}
	}
	\hfill
	\subfloat[a=0.99, b=200 ]{%
	\includegraphics[scale=.29]{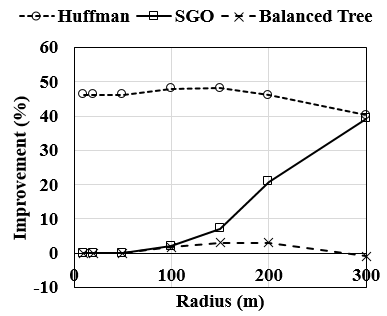}
	}
	\hfill
	\subfloat[a=0.99, b=200 ]{%
	\includegraphics[scale=.29]{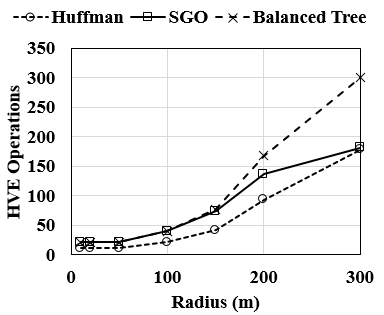}
	}
	\vspace{-10pt}
	\caption{Performance evaluation on synthetic dataset.}
	\label{Fig: Statistical approach}
		\vspace{-10pt}

\end{figure*}


We compare our proposed variable-length encoding scheme with the state-of-the-art fixed-length approach scaled gray optimizer (SGO) from~\cite{shaham2020enhancing}, which uses graph embedding to reflect cell probabilities in the way cell codes are chosen. We also consider as a second benchmark an approach that uses balanced trees, as opposed to Huffman trees.

We use as performance metric the number of HVE bilinear map pairing operations incurred by each technique (which are the most expensive component of the overhead). We present both absolute counts, as well the percentage of {\em improvement} compared to the original fixed-length encoding HVE approach introduced in~\cite{ghinita2014efficient} (which assumes all cells are equally likely to be alerted). 

\subsection{Evaluation on Real Dataset}

Fig.~\ref{Fig: Chicago performance evaluation} shows the performance results obtained on the real dataset. 
The $x$-axis in each graph indicates the size of the alert zone (expressed as radius). For low radii values, the SGO algorithm fails to provide significant improvement, due to the fact that the binary minimization process used by fixed-length encoding approaches is unable to aggregate tokens. In contrast, the proposed variable-length technique using Huffman encoding is able to provide gains of up to 15\% compared to the baseline. In practice, we expect alert zones to be relatively compact compared to the data domain, hence this case is frequently occurring in practice. Furthermore, the results show the superiority of the Huffman code compared to generic variable-length encodings, as the balanced-tree approach benchmark does not produce any improvement.

As the size of alert zone increases, SGO improves, whereas the gain of Huffman encoding decreases. This is expected, since with very large alert zones, it is easy to aggregate tokens, by grouping together cells with low Hamming distance between their codes. However, such an improvement can only be reached when the alert zones are very large, which is not a realistic scenario in practice. In general, the size of alert zones is expected to be small, and their distribution in the data domain sparse, which would further diminish the potential of SGO (and other binary minimization approaches) to produce performance gains, as aggregation requires clustered cells with similar binary codes.

\subsection{Evaluation on Synthetic Dataset}

Performance evaluation results for synthetic data are summarized in Fig.~\ref{Fig: Statistical approach}. We use two inflection points for the sigmoid function $a=0.90,0.99 $, as well as three gradient values $b=10$, $b=100$ and $b=200$.
A similar trend to the real dataset is observed. The Huffman tree approach achieves significantly better performance when the alert zones are compact, which is the expected case in practice.

Two other trends can be observed with respect to the parameters of the sigmoid function. First, a higher inflection point setting results in a more skewed distribution probability on the grid, and leads to a higher performance gain for Huffman encoding compared to competitor approaches. 
The performance gain can be as high as 50\%. This is a positive aspect, since in real life one expects alert cell probabilities to be quite skewed, where more popular areas are visited by more individuals, hence there is more potential for alert events (e.g., public-safety alerts, or visits of a COVID-infected patient to points of interest).
Second, an increase in the gradient of initial probabilities ($b$) also improves the performance gain of Huffman encoding.

We also conducted an experiment under mixed-workload conditions. 
We consider several mixes between short-radius ($20$ meters) and long-radius ($300$ meters) alert zones: {\em W1} (90\% short-10\% long); {\em W2} (75\% short-25\% long);  {\em W3} (25\% short-75\% long); and  {\em W4} (10\% short-90\% long). Results are summarized in Fig.~\ref{Fig: mix}. Our proposed technique outperforms SGO for all considered cases. For mostly-compact alert zones ({\em W1}), the improvement is much higher than that of SGO, with absolute values of up to 40\%.

On the synthetic data, we are also able to perform more in-depth tests where we vary the parameter settings of our proposed approach.
In Fig.~\ref{Fig: Varying grid size}, we vary the grid granularity. The results are obtained for $a=0.95$ and $b=20$. The results show that higher grid granularities lead to higher performance overhead, which is expected, since more cells need to be encoded and encrypted, and thus code lengths increase.
We also observe an interesting trend:
the improvement for a low number of alert cells decreases at higher granularity levels. As the number of grid cells grow, and considering the same sigmoid activation function parameters, there will be more cells with low probabilities of becoming an alert cell. Therefore, the Huffman tree tends to have higher depths. This can be observed more accurately in Fig.~\ref{Fig: Ratio}, where we show the ratio of average length to the maximum length of the Huffman tree for various grid sizes. Hence, the improvement achieved by deterministic minimization lags behind the logic minimization approach, leading to a smaller improvement percentage. 

\begin{figure}[t]
    \centering
	\subfloat[a=0.9, b=100]{%
	\includegraphics[scale=.3]{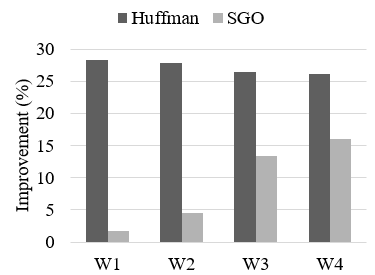}
	}
	\hfill
	\subfloat[a=0.99, b=100]{%
	\includegraphics[scale=.3]{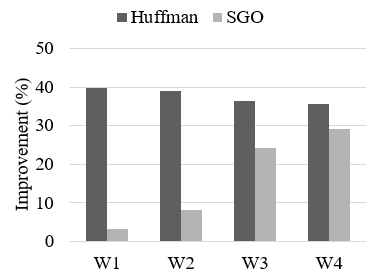}
	}
	\caption{Mixed workloads, synthetic dataset.}
	\vspace{-15pt}
	\label{Fig: mix}
\end{figure}

Finally, we present the run time required to generate indexes and the coding tree in Fig.~\ref{Fig: Time Analysis}. Note that, this is a one-time setup cost, as the process is only run when initializing the system, and has no effect on run-time performance. In the worst case, the process takes minutes for larger-granularity grids.

\begin{figure}[t]
    \centering
	\subfloat[]{%
	\includegraphics[scale=.4]{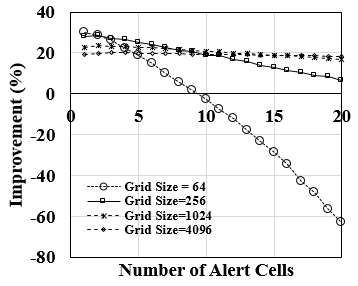}
	}
	\hfill
	\subfloat[]{%
	\includegraphics[scale=.4]{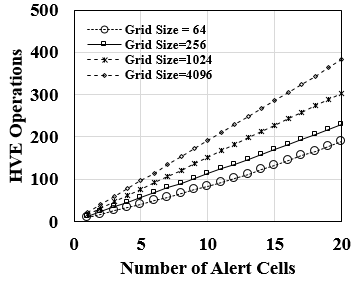}
	}
	\vspace{-5pt}
	\caption{Varying grid granularity, synthetic dataset.}
	\vspace{-10pt}
	\label{Fig: Varying grid size}
\end{figure}

\begin{figure}[t]
\centering
\includegraphics[scale=.5]{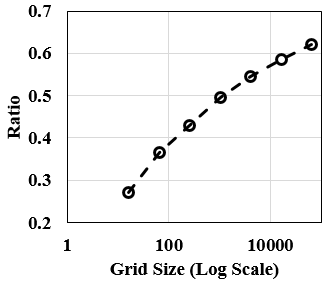}
\hspace{1em}
\centering
\vspace{-10pt}
\caption{Average-to-maximum code length ratio.}
\label{Fig: Ratio}
	\vspace{-15pt}

\end{figure}

\begin{figure}[t]
\centering
\includegraphics[scale=.5]{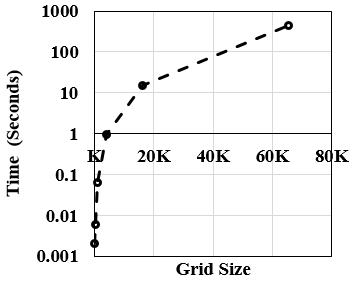}
\hspace{1em}
\centering
\caption{System Initialization Time}
\label{Fig: Time Analysis}
	\vspace{-15pt}

\end{figure}

\section{Related Work}

{\bf Location Privacy.} Early works on location data privacy pivoted around the $k$-anonymity~\cite{sweeney2002k} model. The main idea is to hide users' location among at least k-1 other users to protect user privacy. A preliminary approach to achieve $k$-anonymity was focused on the generation dummy (fake) locations for data points~\cite{kido2005anonymous}. Unfortunately, dummy generation algorithms are shown to be susceptible to inference attacks~\cite{shaham2020privacy}.

An alternative proposed method to achieve $k$-anonymity has been focused on the concept of {\em Cloaking Regions (CRs)}~\cite{gruteser2003anonymous}. Most approaches in this category take advantage of a trusted anonymizer to generate a cluster of $k$ user locations and query the area locations are enclosed by, achieving $k$-anonymity~\cite{gedik2005location,kalnis2007preventing,mokbel2006new}. Approaches based on CRs are effective in a single snapshot~\cite{kalnis2007preventing}; however, once users are considered in trajectories, requiring continuous queries, privacy concerns are posed on the system by inference attacks. Moreover, large CRs are needed in trajectories, significantly reducing the utility of data~\cite{chow2007enabling} as well as posing privacy risks due to inference attacks. The authors in~\cite{damiani2008probe,gruteser2004protecting} aim at providing privacy by distinguishing between sensitive and non-sensitive locations based on user preferences.


{\bf Searchable Encryption}
The main motivation behind searchable encryption techniques is outsourcing the data management to a third party, such as cloud providers without the third party learning about data or queried information by users. The use of a searchable encryption was initially proposed in~\cite{song2000} for a secure cryptographic search of keywords. The approach supports comparison queries~\cite{boneh2006fully} as well as subset queries and conjunctions of equality~\cite{boneh2007conjunctive}. The concept of HVE used in this paper was first proposed in~\cite{boneh2007conjunctive} and later extended in~\cite{blundo2009private}. The authors in~\cite{ghinita2014efficient} proposed the use of HVE to guarantee user privacy in location-based alert systems. Despite promising results of the approach, a major challenge is reducing the computation complexity of HVE at the server where the matching process is conducted. The work in~\cite{shaham2020enhancing} represents the current state-of-the art in location-based alerts with searchable encryption, and it takes into account probabilities of cells being part of an alert zone. A graph embedding technique is used to assign codes to cells in a manner that is aware of their likelihood of becoming alerted. The approach achieves significant improvement in performance compared to \cite{ghinita2014efficient}. However, as our experimental evaluation shows, such improvements are reached only when a relatively large number of alert cells are part of an alert zone. For alert zones with few cells, our approach clearly outperforms that of \cite{shaham2020enhancing}.

\section{Conclusions and Future Work}
We proposed a technique for secure location-based alerts that uses searchable encryption in conjunction with variable-length location encoding. Specifically, using Huffman compression codes, we showed that it is possible to significantly reduce the overhead of searchable encryption for cases where alert zones are compact and sparse, which is the case we believe to be most likely in practice. Extensive analytical and empirical evaluation results prove that our proposed approach significantly outperforms existing fixed-length encoding techniques, with only a small overhead in terms of additional encryption time.

In some cases, our approach may be limited by the lack of a systematic way of obtaining the probability values for various data domain regions. While having accurate probabilities is a plus, we do not require high accuracy in the actual values. In fact, in our design it is often the relative ordering of the probabilities that matters, and not necessarily the exact values.
In practice, one can produce a relatively stable and representative ordering of types of features based on their popularity. Even without precise probability values, one can still obtain significant gains. 

In future work, we plan to investigate more advanced stochastic models that capture correlations between cells in an alert zone, as well as cases when the alert zone evolution over time can be estimated by a spread model (e.g., a chemical gas leak). Significant performance gains can be achieved in such scenarios. One possibility is to model the space and time based on a Markov model. For a grid with $n$ cells, the model would consist of $2^n$ states, each representing a unique subset of grid cells. Next, one can determine a stationary distribution of probabilities over cells, and derive the values required to reach equilibrium. 

Finally, while our work focuses on location data, our design can be extended to benefit other types of data as well. Our assumed semantics for ciphertext processing is that of range queries, and numerous other data types can benefit from secure range queries. However, one has to devise specific encodings and optimizations for each type of data, as straightforward application of HVE to generic data types may lead to high performance overhead, as illustrated in our earlier work~\cite{ghinita2014efficient}.


{\bf Acknowledgment.} This research has been funded in part
by NSF grants IIS-1910950, IIS-1909806 and CNS-2027794,
the USC Integrated Media Systems Center (IMSC), and unrestricted
cash gifts from Google and Microsoft. Any opinions,
findings, and conclusions or recommendations expressed in
this material are those of the author(s) and do not necessarily
reflect the views of any of the sponsors such as the NSF.

\bibliographystyle{abbrv}
\bibliography{sample-base}
 

%
% ...

\end{document}